\documentclass[a4paper]{article}

\advance\textheight by 3cm
\advance\topmargin -2cm
\advance\textwidth by 3cm
\advance\oddsidemargin by -1.5cm
\advance\evensidemargin by -1.5cm

\ifdefined\nopdfsync\else\usepackage{pdfsync}\fi 
\usepackage{amsmath}
\usepackage{mathrsfs}
\usepackage{graphicx}
\usepackage{booktabs}
\usepackage{lscape}
\usepackage[show]{chato-notes}
\usepackage{xspace}
\usepackage{algorithmic}
\usepackage{algorithm}
\usepackage{url}
\usepackage{algorithm, algorithmic}
\floatname{algorithm}{Algorithm}

\newfloat{Algorithm}{htpb}{alg}
\newcounter{prgline}
\newcommand{\pl}{\theprgline\addtocounter{prgline}{1}}
\newcommand{\Var}[1]{\mathrm{Var}[#1]}
\newcommand{\StdDev}[1]{\sqrt{\Var{#1}}}
\newcommand{\Nhat}{\hat N}
\newcommand{\Hhat}{\hat H}
\newcommand{\hhat}{\hat h}

\renewcommand{\epsilon}{\varepsilon}

\newcounter{noqed}
\newcommand{\qed}{ \ifmmode\mbox{ }\fi\rule[-.05em]{.3em}{.7em}\setcounter{noqed}{0}}
\newenvironment{proof}[1][{}]{\noindent{\bf Proof#1. }\setcounter{noqed}{1}}{\ifnum\value{noqed}=1\qed\fi\par\medskip}

\def\..{\,\mathpunct{\ldotp\ldotp}} 
\newcommand{\lst}[2]{${#1}_0$,~${#1}_1$, $\dots\,$,~${#1}_{#2-1}$}

\newcommand{\IF}{\text{\textbf{if}}\xspace}
\newcommand{\BEGIN}{\text{\textbf{begin}}\xspace}

\newcommand{\FOREACH}{\text{\textbf{foreach}}\xspace}

\newcommand{\END}{\text{\textbf{end}}\xspace}
\newcommand{\DO}{\text{\textbf{do}}\xspace}

\newcommand{\RETURN}{\text{\textbf{return}}\xspace}

\newcommand{\UNTIL}{\text{\textbf{until}}\xspace}
\newcommand{\BREAK}{\text{\textbf{break}}\xspace}

\newcommand{\FUNCTION}{\text{\textbf{function}}\xspace}

\newcommand{\COMMENT}{\xspace}

\newcommand{\DD}{\mathscr{D}}

\newtheorem{lemma}{Lemma}
\newtheorem{theorem}{Theorem}
\newtheorem{corollary}{Corollary}

\newcommand{\shiftr}{\mathbin\gg}
\newcommand{\shiftl}{\mathbin\ll}
\newcommand{\band}{\mathbin\&}
\newcommand{\bor}{\mathbin|}
\newcommand{\bxor}{\oplus}
\newcommand{\bnot}[1]{\overline{#1}}

\title{HyperANF: Approximating the Neighbourhood Function of Very Large Graphs
on a Budget} 
\author{Paolo Boldi\quad Marco Rosa \quad Sebastiano
Vigna\\Dipartimento di Scienze dell'Informazione,Universit\`a degli Studi di
Milano, Italy}

\begin{document}
\bibliographystyle{alpha}
\maketitle

\begin{abstract}
The \emph{neighbourhood function} $N_G(t)$ of a graph $G$ gives, for each
$t\in\mathbf N$, the number of pairs of nodes $\langle x, y\rangle$ such that
$y$ is reachable from $x$ in less that $t$ hops. The neighbourhood function
provides a wealth of information about the graph~\cite{PGFANF} (e.g., it easily
allows one to compute its diameter), but it is
very expensive to compute it exactly. Recently, the ANF
algorithm~\cite{PGFANF} (approximate neighbourhood function) has been
proposed with the purpose of approximating $N_G(t)$ on large graphs. We
describe a breakthrough improvement over ANF in terms of
speed and scalability. Our algorithm, called HyperANF, uses the new HyperLogLog
counters~\cite{FFGH} and combines them efficiently through \emph{broadword
programming}~\cite{KnuACPBTT}; our implementation uses
\emph{task decomposition} to exploit multi-core
parallelism. With HyperANF, for the first time we can compute in a few hours
the neighbourhood function of graphs with billions of nodes with a small error and good confidence using a
standard workstation. 

Then, we turn to the study of the distribution of
\emph{distances} between reachable nodes (that can be efficiently
approximated by means of HyperANF), and discover the surprising fact that its
\emph{index of dispersion} provides a clear-cut characterisation of proper social networks vs.~web graphs. We thus propose the \emph{spid}
(Shortest-Paths Index of Dispersion) of a graph as a new, informative statistics
that is able to discriminate between the above two types of
graphs. We believe this is the first proposal of a significant
new non-local structural index for complex networks whose computation is
highly scalable.
\end{abstract}

%
%

\section{Introduction}

The \emph{neighbourhood function} $N_G(t)$ of a graph returns for each
$t\in\mathbf N$ the number of pairs of nodes $\langle x, y\rangle$ such that
$y$ is reachable from $x$ in less that $t$ steps. It provides data about how
fast the ``average ball'' around each node expands. From the neighbourhood
function, several interesting features of a graph can be estimated, and in this
paper we are in particular interested in the \emph{effective diameter}, a
measure of the ``typical'' distance between nodes.

Palmer, Gibbons and Faloutsos~\cite{PGFANF} proposed an
algorithm to \emph{approximate} the neighbourhood function (see their
paper for a review of previous attempts at approximate evaluation);
the authors distribute an associated
tool, \texttt{snap}, which can approximate the neighbourhood function of
medium-sized graphs.
The algorithm keeps track of the number of nodes reachable from each node
using \emph{Flajolet--Martin counters}, a kind of \emph{sketch} that
makes it possible to compute the number of distinct elements of a stream
in very little space. A key observation was that counters associated to
different streams can be quickly combined into a single counter associated to
the concatenation of the original streams. 

In this paper, we describe HyperANF---a breakthrough improvement over ANF in
terms of speed and scalability. HyperANF uses the new HyperLogLog counters~\cite{FFGH},
and combines them efficiently by means of \emph{broadword
programming}~\cite{KnuACPBTT}. Each counter is made by a number of \emph{registers}, and the number of registers
depends only on the required precision. 
The size of each register is \emph{doubly logarithmic} in the number of nodes of
the graph, so HyperANF, for a fixed precision, scales almost linearly in memory (i.e., $O(n\log \log n)$).
By contrast, ANF memory requirement is $O(n\log n)$.	

Using HyperANF, for the first time we can compute in a few hours the
neighbourhood function of graphs with more than one billion nodes with a small error and good
confidence using a standard workstation with 128\,GB of RAM.  Our algorithms are implement in a tool distributed
as free software within the WebGraph
framework.\footnote{See~\cite{BoVWFI}. \texttt{http://webgraph.dsi.unimi.it/}.}

Armed with our tool, we study several datasets, spanning from small social
networks to very large web graphs. We isolate a statistically defined feature,
the \emph{index of dispersion of the distance distribution}, and show that
it is able to tell ``proper" social networks from web graphs in a natural way.

\section{Related work}

HyperANF is an evolution of ANF~\cite{PGFANF}, which is implemented by the tool
\texttt{snap}. We will give some timing comparison with \texttt{snap}, but we
can only do it for relatively small networks, as the large memory footprint of
\texttt{snap} precludes application to large graphs.

Recently, a MapReduce-based distributed implementation of ANF called
HADI~\cite{KTAHMRLG} has been presented. HADI runs on one of the fifty largest
supercomputers---the Hadoop cluster M45. The only published data about HADI's
performance is the computation of the neighbourhood function of a Kronecker
graph with 2 billion links, which required half an hour using 90 machines.
HyperANF can compute the same function in \emph{less than fifteen minutes on a
laptop}.

The rather complete survey of related literature in~\cite{KTAHMRLG} shows that
essentially no data mining tool was able before ANF to approximate the
neighbourhood function of very large graphs reliably. A remarkable exception is
Cohen's work~\cite{CohSEFATCR}, which provides strong theoretical guarantees but
experimentally turns out to be not as scalable as the ANF approach; it is worth
noting, though, that one of the proposed applications of~\cite{CohSEFATCR} (\emph{On-line estimation
of weights of growing sets}) is structurally identical to ANF.

All other results published before ANF relied on a small number of breadth-first
visits on uniformly sampled nodes---a process that has no provable statistical accuracy or
precision. Thus, in the rest of the paper we will compare experimental data with
\texttt{snap} and with the published data about HADI.

\section{HyperANF}

In this section, we present the HyperANF algorithm for computing an
approximation of the neighbourhood function of a graph; we start by recalling
from~\cite{FFGH} the notion of HyperLogLog counter upon which our algorithm
relies. We then describe the algorithm, discuss how it can be implemented to be
run quickly using broadword programming and task decomposition, and give
results about its memory requirements and precision.

\subsection{HyperLogLog counters}
\label{sec:hyper}

\emph{HyperLogLog counters}, as described in~\cite{FFGH} (which is based
on~\cite{DuFLCLC}), are used to count approximately the number of distinct
elements in a stream. For the purposes of the present paper, we need to recall briefly their behaviour. 
Essentially, these probabilistic counters are a sort of \emph{approximate set
representation} to which, however, we are only allowed to pose questions about
the (approximate) size of the set.\footnote{We remark that in principle
$O(\log n)$ bits are necessary to estimate the number of unique elements in a
stream~\cite{AMSSCAFM}. HyperLogLog is a practical counter that starts from the
assumption that a hash function can be used to turn a stream into an
\emph{idealised multiset} (see~\cite{FFGH}).}

Let $\DD$ be a fixed domain and $h: \DD \to 2^\infty$ be a hash function mapping
each element of $\DD$ into an infinite binary sequence. The function is fixed with the
only assumption that ``bits of hashed values are assumed to be independent and
to have each probability $\frac{1}{2}$ of occurring''~\cite{FFGH}. 

For a given $x \in 2^\infty$, let $h_t(x)$ denote the sequence made by the
leftmost $t$ bits of $h(x)$, and $h^t(x)$ be the sequence of remaining bits of
$x$; $h_t$ is identified with its corresponding integer
value in the range $\{\,0,1,\dots,2^t-1\,\}$. Moreover,
given a binary sequence $w$, we let $\rho^+(w)$ be the number of leading zeroes
in $w$ plus one\footnote{We remark that in the original HyperLogLog papers
$\rho$ is used to denote $\rho^+$, but $\rho$ is a somewhat standard notation
for the ruler function~\cite{KnuACPBTT}.} (e.g., $\rho^+(00101)=3$). Unless
otherwise specified, all logarithms are in base 2.

%
\begin{algorithm}
\begin{tabbing}
\setcounter{prgline}{0}
\hspace{0.5cm} \= \hspace{0.3cm} \= \hspace{0.3cm} \= \hspace{0.3cm} \=
\hspace{0.3cm} \= \hspace{0.3cm} \=\kill\\
\pl\>$h: \DD \to 2^\infty$, a hash function from the domain of items\\
\pl\>$M[-]$ the counter, an array of $m=2^b$ registers\\
\pl\>\>(indexed from 0) and set to $-\infty$\\ 
\pl\>\\
\pl\>\FUNCTION $\mathrm{add}(\text{$M$: counter},\text{$x$: item})$ \\
\pl\>\BEGIN\\
\pl\>\>$i\leftarrow h_b(x)$;\\
\pl\>\>$M[i]\leftarrow\max \bigl\{M[i], \rho^+\bigl(h^b(x)\bigr)\bigr\}$\\
\pl\>\END; \COMMENT{// function add}\\
\pl\>\\ 
\pl\>\FUNCTION $\mathrm{size}(\text{$M$: counter})$ \\
\pl\>\BEGIN\\
\pl\>\>$Z\leftarrow \left(\sum_{j=0}^{m-1} 2^{-M[j]}\right)^{-1}$;\\
\pl\>\>\RETURN $E=\alpha_m m^2 Z$\\
\pl\>\END; \COMMENT{// function size}\\
\pl\>\\ 
\pl\>\FOREACH item $x$ seen in the stream \BEGIN\\ 
\pl\>\>add($M$,$x$)\\
\pl\>\END;\\
\pl\>print $\mathrm{size}(M)$\\
\end{tabbing}
\caption{\label{algo:Hyperloglog}The Hyperloglog counter as described
in~\cite{FFGH}: it allows one to count (approximately) the number of distinct
elements in a stream. $\alpha_m$ is a constant whose value depends on $m$ and
is provided in~\cite{FFGH}. Some technical details have been simplified.}
\end{algorithm}

The value $E$ printed by Algorithm~\ref{algo:Hyperloglog}
is~\cite{FFGH}[Theo\-rem 1] an asymptotically almost unbiased estimator for the
number $n$ of distinct elements in the stream; for $n \to\infty$, the 
\emph{relative standard deviation} (that
is, the ratio between the standard deviation of $E$ and $n$) is at most
$\beta_m/\sqrt m\leq 1.06/\sqrt{m}$, where $\beta_m$ is a suitable constant
(given in~\cite{FFGH}). Moreover~\cite{DuFLCLC} even if the size of the registers (and of the hash function) used by the algorithm is
unbounded, one can limit it to $\log\log(n/m)+\omega(n)$ bits obtaining almost
certainly the same output ($\omega(n)$ is a function going to infinity
arbitrarily slowly); overall, the algorithm requires $(1+o(1)) \cdot m
\log\log(n/m)$ bits of space (this is the reason why these counters are called
HyperLogLog). Here and in the rest of the paper we tacitly assume that $m\geq
64$ and that registers are made of $\lceil \log\log n\rceil$ bits.

\subsection{The HyperANF algorithm}

The approximate neighbourhood function algorithm described in~\cite{PGFANF}
is based on the observation that $B(x,r)$, the ball of radius $r$ around node $x$, satisfies
\[
B(x,r) = \bigcup_{x\to y}B(y,r-1).
\]
Since $B(x,0)=\{\,x\,\}$, we can compute each $B(x,r)$ incrementally using
sequential scans of the graph (i.e., scans in which we go in turn through the
successor list of each node). The obvious problem is that during the scan we
need to access randomly the sets $B(x,r-1)$ (the sets $B(x,r)$ can be just
saved on disk on a \emph{update file} and reloaded later). Here
probabilistic counters come into play; to be able to use them, though, we need 
to endow counters with a primitive for the union.
Union can be implemented provided that the counter associated to
the stream of data $AB$ can be computed from the counters associated to $A$ and $B$; in the case of
HyperLogLog counters, this is easily seen to correspond to maximising the
two counters, register by register.

The observations above result in Algorithm~\ref{algo:HyperANF}: the algorithm
keeps one HyperLogLog counter for each node; at the $t$-th iteration of the main loop, 
the counter $c[v]$ is in the same state as if it would have been fed with $B(v,t)$, 
and so its expected value is $|B(v,t)|$. As a result, the sum of all $c[v]$'s is
an (almost) unbiased estimator of $N_G(t)$ (for a precise statement, see
Theorem~\ref{teo:corrhyperANF}).

\begin{algorithm}
\begin{tabbing}
\setcounter{prgline}{0}
\hspace{0.5cm} \= \hspace{0.3cm} \= \hspace{0.3cm} \= \hspace{0.3cm} \=
\hspace{0.3cm} \= \hspace{0.3cm} \=\kill\\
\pl\>$c[-]$, an array of $n$ HyperLogLog counters\\
\pl\>\\ 
\pl\>\FUNCTION $\mathrm{union}(\text{$M$: counter},\text{$N$: counter})$ \\
\pl\>\>\FOREACH $i<m$ \BEGIN\\
\pl\>\>\>$M[i] \leftarrow \max(M[i],N[i])$\\
\pl\>\>\END\\
\pl\>\END; \COMMENT{// function union}\\
\pl\>\\ 
\pl\>\FOREACH $v\in n$ \BEGIN\\
\pl\>\>add $v$ to $c[v]$\\
\pl\>\END;\\ 	
\pl\>$t\leftarrow 0$;\\
\pl\>\DO \BEGIN\\
\pl\>\>$s\leftarrow \sum_v \mathrm{size}(c[v])$;\\
\pl\>\>Print $s$ (the neighbourhood function $N_G(t)$)\\
\pl\>\>\FOREACH $v\in n$ \BEGIN\\
\pl\>\>\>$m\leftarrow c[v]$;\\
\pl\>\>\>\FOREACH $v\rightarrow w$ \BEGIN\\
\pl\>\>\>\>$m\leftarrow \mathrm{union}(c[w],m)$\\
\pl\>\>\>\END; 	\\
\pl\>\>\>write $\langle v,m\rangle$ to disk\\
\pl\>\>\END; 	\\
\pl\>\>Read the pairs $\langle v,m\rangle$ and update the array $c[-]$\\
\pl\>\>$t\leftarrow t+1$\\
\pl\>\UNTIL no counter changes its value.
\end{tabbing}
\caption{\label{algo:HyperANF}The basic HyperANF algorithm in pseudocode. The
algorithm uses, for each node $i\in n$, an initially empty HyperLogLog counter
$c_i$. The function $\mathrm{union}(-,-)$ maximises two counters register by
register.}
\end{algorithm}

We remark that the only sound way of running HyperANF (or ANF) is to wait for
all counters to stabilise (e.g., the last iteration must leave all counters
unchanged). As we will see, any alternative termination condition may lead to
arbitrarily large mistakes on pathological graphs.\footnote{We remark that
\texttt{snap} uses a threshold over the relative increment in the number of
reachable pairs as a termination condition, but this trick makes the tail of
the function unreliable.}

\subsection{HyperANF at hyper speed}
\label{sec:broad}

Up to now, HyperANF has been described just as ANF with HyperLogLog counters.
The effect of this change is an exponential reduction in the memory footprint
and, consequently, in memory access time. We now describe the the algorithmic
and engineering ideas that made HyperANF much faster, actually so fast that
it is possible to run it up to stabilisation.

\smallskip\noindent\textbf{Union via broadword programming.}
Given two HyperLogLog counters that have been set by streams $A$ and $B$, the
counter associated to the stream $AB$ can be build by maximising in parallel
the registers of each counter. That is, the register $i$ of the new counter is
given by the maximum between the $i$-th register of the first counter and
the $i$-th register of the second counter.

Each time we scan a successor list, we need to maximise a large number of
registers and store the resulting counter. The immediate way of obtaining this
result requires extracting the value of each register, maximise it with the
other corresponding registers, and writing down the result in a temporary
counter. This process is extremely slow, as registers are packed in 64-bit
memory words. In the case of Flajolet--Martin counters, the problem is easily
solved by computing the logical OR of the words containing the registers. In
our case, we resort to \emph{broadword programming} techniques. If the machine
word is $w$, we assume that at least $w$ registers are allocated to each
counter, so each set of registers is word-aligned. 

Let $\shiftr$ and
$\shiftl$ denote right and left (zero-filled) shifting, $\band$,
$\bor$ and $\bxor$ denote bit-by-bit not, and, or, and xor; $\bnot x$ denotes
the bit-by-bit complement of $x$.


We use $L_k$ to denote the constant whose ones are in position $0$, $k$,
$2k$, \ldots\, that is, the constant with the \emph{lowest} bit of each $k$-bit subword
set (e.g, $L_8=0x01010101010101010101$). We use $H_k$ to denote $L_k \shiftl
k-1$, that is, the constant with the \emph{highest} bit of each $k$-bit subword
set (e.g, $H_8=0x8080808080808080$).

It is known (see~\cite{KnuACPBTT}, or \cite{VigBIRSQ} for an elementary proof),
that the following expression
\[
x<_k^u y :=\Bigl(\bigl(\, ( (x \bor H_k ) - ( y\band \bnot{H_k} ) )\bor {x\oplus y}\bigr)
\oplus (x\bor \bnot y)\Bigr)\band H_k.
\]
performs a parallel unsigned comparison $k$-by-$k$-bit-wise. At the end of the
computation, the highest bit of each block of $k$ bits will be set iff the
corresponding comparison is true (i.e., the value of the block in $x$ is
strictly smaller than the value of the block in $y$).

Once we have computed $x<_k^u$, we generate a mask that is made
entirely of 1s, or of 0s, for each $k$-bit block, depending on whether we should select the value of $x$ or $y$
for that block:
\[
		 m = \biggl( \Bigl( \bigl( (x<_k^u y)\shiftr k-1 \bor H_k \bigr) - L_k \Bigr)
		 \bor H_k \biggr) \bxor (x<_k^u y)
\]
This formula works by moving the high bit denoting the result of the comparison
to the least significant bit (of each $k$-bit block). Then, we or with $H_k$ and
subtract $1$ from each block, obtaining either a mask with just the high bit
set (if we were starting from 1) or a mask with all bits sets except for the
high bit (if we were starting from 0). The last two operation fix those values
so that they become $00\cdots0$ or $11\cdots 1$. The result of the maximisation
process is now just $x \band m \bor y \band \bnot m$.

This discussion assumed that the set of registers of a counter is stored in a
single machine word. In a realistic setting, the registers are spread among
several consecutive words, and we use multiple precision subtractions and shifts
to apply the expressions above on a sequence of words. All other (logical)
operations have just to be applied to each word in sequence.

All in all, by using the techniques above we can improve the speed of
maximisation by a factor of $w/k$, which in our case is about 13 (for graphs of
up to $2^{32}$ nodes). This actually results in a sixfold speed improvement of
the overall application in typical cases (e.g., web graphs and $b=8$), as about
90\% of the computation time is spent in maximisation.

\smallskip\noindent\textbf{Parallelisation via task decomposition.} Although HyperANF is
written as a sequential algorithm, the outer loop lends itself to
be executed in parallel, which can be extremely fruitful on a modern multicore
architecture; in particular, we approach this idea using
\emph{task decomposition}. We divide the iteration on the whole set of nodes into
a set of small tasks (in the order of the thousands), where each task consists
in iterating on a contiguous segment of nodes. A pool of threads picks up
the first available task and solves it: as a result, we obtain a performance
improvement that is linear in the number of cores. Threads
can be designed to be extremely agile, helped by WebGraph's facilities which
allow us to provide each thread with a lightweight copy of the graph that shares
the bitstream and associated information with all other threads.

\smallskip\noindent\textbf{Tracking modified counters.}
It is an easy observation that a counter $c$ that does not change its value is
not useful for the next step of the computation: all counters using $c$ during
their update would not change their value when maximising with $c$ (and
we do not even need to write $c$ on disk). We
thus keep track of modified counters and skip altogether the maximisation step with
unmodified ones. Since, as we already remarked, 90\% of computation time is
spent in maximisation, this approach leads to a large speedup after the first
phases of the computation, when most counters are stabilised.

For the same reason, we keep track of the harmonic partial sums of small blocks
(e.g., $64$) of counters. The amount of memory required is negligible, but if no
counter in the block has been modified, we can avoid a costly computation.

\smallskip\noindent\textbf{Systolic computation.}
HyperANF can be run in \emph{systolic} mode. In this case, we use also the
transposed graph: whenever a counter changes, it \emph{signals} back to its
predecessors that at the next round they could change their values. Now, at each
iteration nodes that have not been signalled are entirely skipped during the
computation. Systolic computations are fundamental to get high-precision runs,
as they reduce the cost of an iteration to scanning only the arcs of the graph
that are actually moving information around. We switch to systolic computation
when less than one quarter of the counters change their values.

\subsection{Correctness, errors and memory usage}

Very little has been published about the statistical behaviour of ANF. The
statistical properties of approximate counters are well known, but the values
of such counters for each node are \emph{highly dependent}, and adding them in a
large amount can in principle lead to an arbitrarily large variance. Thus,
making precise statistical statements about the outcome of a computation of ANF or HyperANF requires some
care. The discussion in the following sections is based on HyperANF, but its
results can be applied \textit{mutatis mutandis} to ANF as well.

Consider the output $\Nhat_G(t)$ of algorithm
\ref{algo:HyperANF} at a fixed iteration $t$. We can see it as a random variable 
\[
\Nhat_G(t)=\sum_{i \in n} X_{i,t}
\] 
where\footnote{Throughout this paper, we use von Neumann's
notation $n=\{\,0,1,\dots,n-1\,\}$, so $i\in n$ means that $0\leq i<n$.} each
$X_{i,t}$ is the HyperLogLog counter that counts nodes reached by node $i$ in $t$ steps; what we want to prove in this section is a bound on
the relative standard deviation of $\Nhat_G(t)$ (such a proof,
albeit not difficult, is not provided in the papers about ANF). First observe that~\cite{FFGH}, for a fixed a number of registers $m$ per counter, 
the standard deviation of $X_{i,t}$ satisfies
\[\frac{\StdDev{X_{i,t}}}{|B(i,t)|}\leq \eta_m, \] where $\eta_m$ is the guaranteed relative standard deviation of a HyperLogLog counter. Using
the subadditivity of standard deviation (i.e., if $A$ and $B$ have finite variance, $\StdDev{A+B}\leq \StdDev A+\StdDev B$), 
we prove the following
\begin{theorem}\label{teo:corrhyperANF}
The output $\Nhat_G(t)$ of Algorithm~\ref{algo:HyperANF} at the $t$-th iteration
is an asymptotically almost unbiased estimator\footnote{From now on, for the
sake of readability we shall ignore the negligible bias on $\Nhat_G(t)$ as an estimator for $N_G(t)$: the other estimators that will
appear later on will be qualified as ``(almost) unbiased'', where ``almost''
refers precisely to the above mentioned negligible bias.} of $N_G(t)$, that is
\[
	\frac{E[\Nhat_G(t)]}{N_G(t)}=1+\delta_1(n)+o(1) \text{ for $n\to \infty$},
\]
where $\delta_1$ is the same as in~\cite{FFGH}[Theorem 1] (and
$|\delta_1(x)|<5\cdot 10^{-5}$ as soon as $m\geq 16$). 

Moreover, $\Nhat_G(t)$
has the same relative standard deviation of the $X_i$'s, that is
\[
\frac{\StdDev{\Nhat_G(t)}}{N_G(t)} \leq \eta_m.
\]
\end{theorem}
\begin{proof}
We have that $E[\Nhat_G(t)]=E\bigl[\sum_{i \in n}X_{i,t}\bigr]$.
By Theorem 1 of~\cite{FFGH},
$E[X_{i,t}]=|B(i,t)|\left(1+\delta_1(n)+o(1)\right)$, hence the first statement.
For the second result, we have:
\[\frac{\StdDev {\Nhat_G(t)}}{N_G(t)} \leq \frac{\sum_{i\in
n}\StdDev{X_i}}{N_G(t)}\leq \frac{\eta_m\sum_{i\in n}|B(i,t)|}{N_G(t)} =
\eta_m.
\]
\end{proof}

Since, as we recalled in Section~\ref{sec:hyper}, the relative standard
deviation $\eta_m$ satisfies $\eta_m\leq 1.06/\sqrt m$, 
to get a specific value $\eta$ it is sufficient to choose $m\approx
1.12/\eta^2$; this assumption yields an overall space requirement of about
\[
	\frac{1.12}{\eta^2} n \log\log n \qquad\text{bits}
\]
(here, we used the obvious upper bound $|B(i,t)|\leq n$).
For instance, to obtain a relative standard deviation of $9.37\%$ (in every
iteration) on a graph of one billion nodes one needs $74.5$\,GB of main memory
for the registers (for a comparison, \texttt{snap} would require $550$\,GB). Note that
since we write to disk the new values of the registers, this is actually the
only significant memory requirement (the graph can be kept on disk and
mapped in memory, as it is scanned almost sequentially).


\medskip
Applying Chebyshev's inequality, we obtain the following: 
\begin{corollary}\label{cor:cheb}
For every $\epsilon$, 
\[
\Pr\left[\frac{\Nhat_G(t)}{N_G(t)} \in (1-\epsilon,1+\epsilon)\right] \geq
1-\frac{\eta_m^2}{\epsilon^2}.
\]
\end{corollary}

In~\cite{FFGH} it is argued that the HyperLogLog error is approximately
Gaussian; the counters, however, are \emph{not} statistically independent and
in fact the overall error does not appear to be normally distributed.
Nonetheless, for every fixed $t$, the random variable $\Nhat_G(t)$ seems to be
unimodal (for example, the average p-value of the Dip unimodality
test~\cite{HHDTU} for the \texttt{cnr-2000} dataset is $0.011$), so we can apply
the Vysochanski{\u\i}-Petunin inequality~\cite{VPJTSRUD}, obtaining the bound
\[
\Pr\left[\frac{\Nhat_G(t)}{N_G(t)} \in (1-\epsilon,1+\epsilon)\right] \geq
1-\frac{4\eta_m^2}{9\epsilon^2} .
\]
In the rest of the paper, to state clearly our theorems we will always assume
error $\epsilon$ with confidence $1-\delta$. It is useful, as a practical
reminder, to note that because of the above inequality for each
point of the neighbourhood function we can assume a relative error of $k\eta_m$ with confidence $1-4/(9k^2)$
(e.g., $2\eta_m$ with $90$\% confidence, or $3\eta_m$ with $95$\% confidence).

As an empirical counterpart to the previous results, we considered a
relatively small graph of about $325\,000$ nodes (\texttt{cnr-2000}, see
Section~\ref{sec:experiments} for a full description) for which we can compute
the exact neighbourhood function $N_G(-)$; we ran HyperANF 500 times with
$m=256$.
 At least $96$\% of the samples (for all $t$) has a
relative error smaller than twice the theoretical relative standard deviation
$6.62\%$. The percentage jumps up to $100$\% for three times the relative
standard deviation, showing that the distribution of the values behaves better than
what the theory would guarantee.

\section{Deriving useful data}

As advocated in~\cite{PGFANF}, being able to estimate the neighbourhood function
on real-world networks has several interesting applications. Unfortunately, all
published results we are aware of lack statistical satellite data
(such as confidence intervals, or distribution of the computed values) that make
it possible to compare results from different research groups. Thus, in this
section we try to discuss in detail how to derive useful data from an approximation of
the neighbourhood function.

\smallskip\noindent\textbf{The distance cdf.} We start from the apparently easy task
of computing the \emph{cumulative distribution function of distances} of the graph $G$ (in short, \emph{distance cdf}), which is
the function $H_G(t)$ that gives the fraction of reachable pairs at distance at
most $t$, that is,
\[
	H_G(t) = \frac{N_G(t)}{\max_t N_G(t)}.
\] 
In other words, given an exact computation of
the neighbourhood function, the distance cdf can be easily obtained by
dividing all values by the largest one. 
Being able to estimate $N_G(t)$ allows one to produce a reliable approximation
of the distance cdf:
\begin{theorem}\label{thm:hhat}
Assume $N_G(t)$ is known for each $t$ with error $\epsilon$ and
confidence $1-\delta$, that is
\[
\Pr\left[\frac{\Nhat_G(t)}{N_G(t)} \in (1-\epsilon,1+\epsilon)\right]\geq
1-\delta.
\]
Let $\Hhat_G(t)= \Nhat_G(t)/\max_t \Nhat_G(t)$.
Then $\Hhat_G(t)$ is an (almost) unbiased estimator for $H_G(t)$; moreover, 
for a fixed sequence \lst tk, for every $\epsilon$ and all $0\leq i<k$ we have
that $\Hhat_G(t_k)$ is known with error $2\epsilon$ and confidence
$1-(k+1)\delta$, that is,
\[
\Pr\left[\bigwedge_{i\in k}\frac{\Hhat_G(t_i)}{H_G(t_i)} \in
(1-2\epsilon,1+2\epsilon)\right] \geq 1 - (k+1)\delta.
\]
\end{theorem}
\begin{proof}
Note that if 
\[
1-\epsilon\leq \Nhat_G(t)/N_G(t) \leq 1+\epsilon
\] 
holds for every $t$, then \emph{a fortiori} 
\[1-\epsilon\leq \max_t\Nhat_G(t)/\max_tN_G(t) \leq 1+\epsilon
\] 
(because, although the maxima might be first attained at different values of
$t$, the same holds for any larger values).
As a consequence,
 \[1-2\epsilon\leq \frac{1-\epsilon}{1+\epsilon} \leq
\frac{\Hhat_G(t)}{H_G(t)} \leq \frac{1+\epsilon}{1-\epsilon} \leq 1+2\epsilon.\]
The probability $1-(k+1)\delta$ is immediate from the union bound, as we are
considering $k+1$ events at the same time.
\end{proof} 
Note two significant limitations: first of all, making precise statements
(i.e., with confidence) about \emph{all} points of $H_G(t)$ requires a very high
initial error and confidence. Second,
the theorem holds if HyperANF has been run up to stabilisation, so that the probabilistic guarantees of HyperLogLog hold for all
$t$.

The first limitation makes in practice impossible to get directly sensible
confidence intervals, for instance, for the average distance or
higher moments of the distribution (we will elaborate further on this point later). Thus, only
statements about a small, finite number of points can be approached directly.

The second limitation is somewhat more serious in theory, albeit in practice it
can be circumvented making suitable assumptions about the graph under
examination (which however should be clearly stated along the data). Consider
the graph $G$ made by two $k$-cliques joined by a unidirectional path of $\ell$
nodes (see Figure~\ref{fig:counterexdiam}). Even neglecting the effect of
approximation, $G$ can ``fool'' HyperANF (or ANF) so that the distance cdf
is completely wrong (see Figure~\ref{fig:countercdf}) when using \emph{any}
stopping criterion that is not stabilisation.

\begin{figure}[htb]
\centering
\includegraphics[scale=.7]{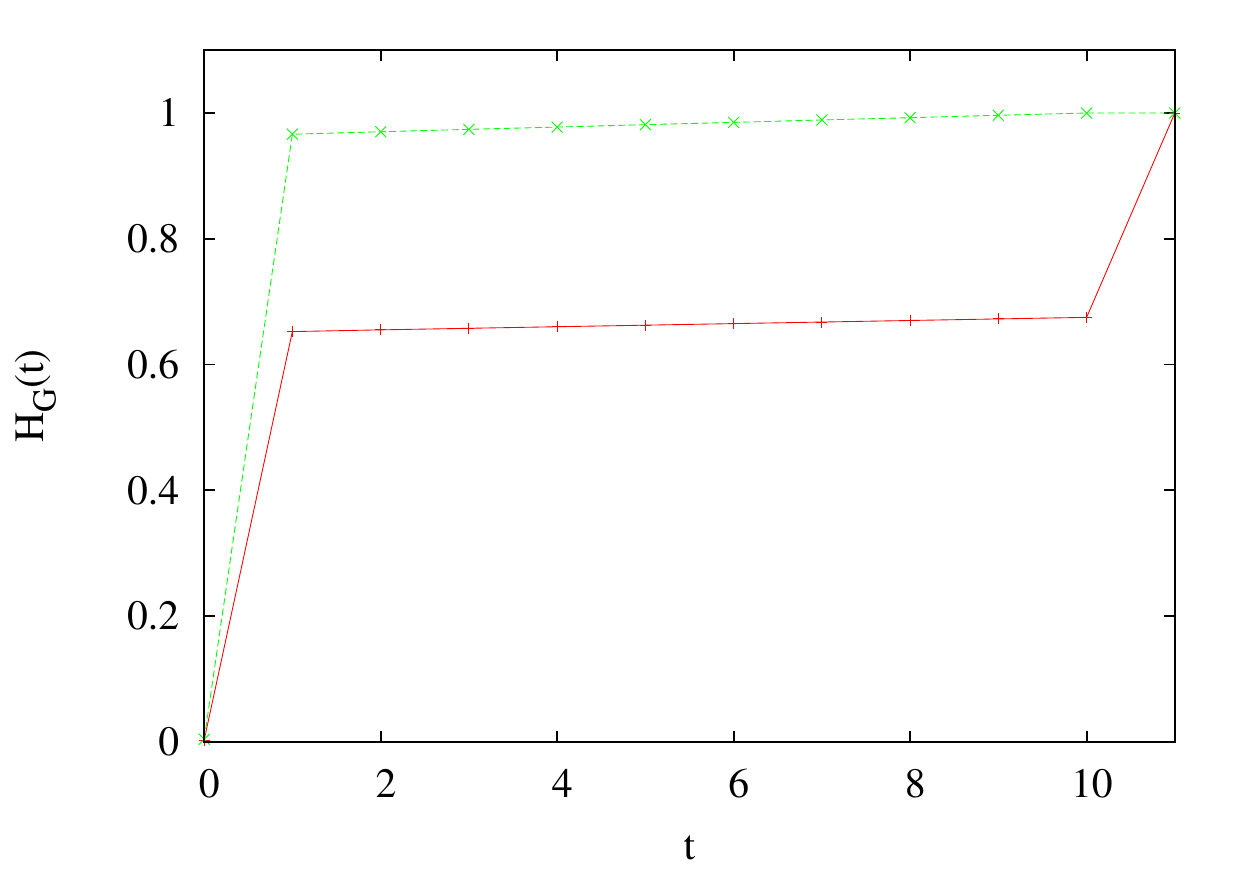}
\caption{\label{fig:countercdf}The real cdf of the graph in
Figure~\ref{fig:counterexdiam} (+), and the one that would be computed using
\emph{any} termination condition that is not stabilisation (*); here $\ell=10$
and $k=260$.}
\end{figure}

\begin{figure}[htb]
\centering
\includegraphics[scale=.7]{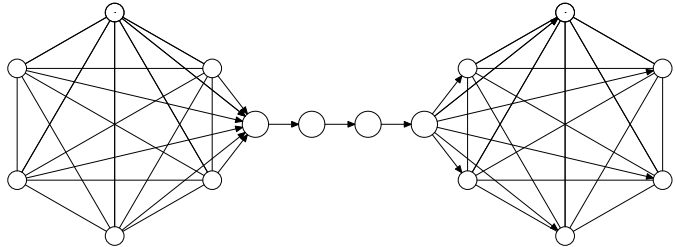}
\caption{\label{fig:counterexdiam}Two $k$-cliques joined by a
unidirectional path of $\ell$ nodes: terminating even one step earlier than 
stabilisation completely miscalculates the distance cdf (see
Figure~\ref{fig:countercdf}); the effective diameter is $\ell+1$, but
terminating even just one step earlier than stabilisation yields an estimated
effective diameter of 1.}
\end{figure}

Indeed, the exact neighbourhood function of $G$ is given by:
\[
	N_G(t)=\begin{cases}
			2k+\ell \qquad \text{if $t=0$}\\
			(t+1)\left(2k+\ell-\frac t2\right)-2k+2k^2 & \text{if $1\leq t\leq
			\ell$}\\
			(\ell+1)\left(2k+\frac \ell2\right)-2k+3k^2 &\text{if $\ell<t$.}
			\end{cases}
\] 
The key observation is that the very last value is significantly larger than
all previous values, as at the last step the nodes of the right clique become
reachable from the nodes of the first clique. Thus, if iteration stops
before stabilisation,\footnote{We remark that stabilisation can occur, in
principle, even before the last step because of hash collisions in HyperLogLog
counters, but this will happen with a controlled probability.} the normalisation
factor used to compute the cdf will be smaller by $\approx k^2$ than the actual
value, causing a completely wrong estimation of the cdf, as shown in Figure~\ref{fig:countercdf}.

Although this counterexample (which can be easily a\-dap\-ted to be symmetric)
is definitely pathological, it suggests that a particular care should be taken when
handling graphs that present narrow ``tubes'' connecting large connected
components: in such scenarios, the function $N_G(t)$ exhibits relatively long
plateaux (preceded and followed by sharp bumps) that may fool the computation of
the cdf. 

\smallskip\noindent\textbf{The effective diameter.} The first application of ANF
was the computation of the \emph{effective diameter}. The effective diameter of
$G$ at $\alpha$ is the smallest $t_0$ such that $H_G(t_0)\geq \alpha$; when $\alpha$
is omitted, it is assumed to be $\alpha=.9$.\footnote{The actual diameter of $G$
is its effective diameter at $1$, albeit the latter is defined for all graphs
whereas the former makes sense only in the strongly connected case.} The \emph{interpolated} effective
diameter is obtained in the same way on the \emph{linear interpolation}
of the points of the neighbourhood function.

Since that the function $\Hhat_G(t)$ is necessarily monotone in $t$
(independently of the approximation error), from Theorem~\ref{thm:hhat} we
obtain:
\begin{corollary}
Assume $\Nhat_G(t)$ is known for each $t$ with error $\epsilon$ and
confidence $1-\delta$, and there are points $s$ and $t$ such that
\[\frac{\Hhat_G(s)}{1-2\epsilon} \leq \alpha \leq
\frac{\Hhat_G(t)}{1+2\epsilon}.\] 
Then, with probability $1-3\delta$ the
effective diameter at $\alpha$ lies in $[s\.. t]$.
\end{corollary}
Unfortunately, since the effective diameter depends sensitively on the
distance cdf, again termination conditions can produce arbitrary errors.
Getting back to the example of Figure~\ref{fig:counterexdiam}, with a
sufficiently large $k$, for example $k=2\ell^2+5\ell+2$, the effective diameter
is $\ell+1$, which would be correctly output after $\ell+1$ iterations,
whereas even stopping one step earlier (i.e., with $t=\ell$) 
would produce $1$ as output, yielding an arbitrarily large error. 
\texttt{snap}, indeed, fails to produce the correct result on
this graph, because it stops iterating whenever the ratio between
two successive iterates of $N_G$ is sufficiently close to 1.  

\begin{algorithm}
\begin{tabbing}
\setcounter{prgline}{0}
\hspace{0.5cm} \= \hspace{0.3cm} \= \hspace{0.3cm} \= \hspace{0.3cm} \=
\hspace{0.3cm} \= \hspace{0.3cm} \=\kill\\
\pl\>\FOREACH $t=0,1,\dots$ \BEGIN\\
\pl\>\>compute $\Nhat_G(t)$ (error $\epsilon$, confidence $1-\delta$)\\
\pl\>\>\IF (some termination condition holds) \BREAK\\
\pl\>\END;\\
\pl\>$M\leftarrow \max \Nhat_G(t)$\\
\pl\>find the largest $D^-$ such that
$\Nhat_G(D^-)/M\leq \alpha(1-2\epsilon)$\\ 
\pl\>find the smallest $D^+$ such that
$\Nhat_G(D^+)/M\geq \alpha(1+2\epsilon)$\\ 
\pl\>output $[D^-\..D^+]$ with confidence $1-3\delta$\\
\pl\>\END;
\end{tabbing}
\caption{\label{algo:diameter}Computing the effective diameter at $\alpha$ of a
graph $G$; Algorithm~\ref{algo:HyperANF} is used to compute $\Nhat_G$.}
\end{algorithm}

Algorithm~\ref{algo:diameter} is used to estimate the effective
diameter of a graph; albeit this approach is reasonable (and actually it is
similar to that adopted by \texttt{snap}, although the latter does not provide any
confidence interval), unless the neighbourhood function is known with very high 
precision it is almost impossible to obtain good upper bounds, because of the
typical flatness of the distance cdf after the 90th percentile. Moreover,
results computed using a termination condition different from stabilisation
should always be taken with a grain of salt because of the discussion above.

\smallskip\noindent\textbf{The distance density function.} The situation, from a
theoretical viewpoint, is somehow even worse when we consider the density
function $h_G(-)$ associated to the cdf $H_G(-)$. Controlling the error on
$h_G(-)$ is not easy:
\begin{lemma}
\label{lemma:df}
Assume that, for a given $t$, $\Hhat_G(t)$ is an estimator of $H_G(t)$ with
error $\epsilon$ and confidence $1-\delta$. 
Then $\hhat_G(t)= h_G(t)\pm 2\epsilon$ with confidence $1-2\delta$.
\end{lemma}
\begin{proof}
With confidence
$1-2\delta$,
\begin{align*}
\hhat_G(t)&=\Hhat_G(t)-\Hhat_G(t-1)\\
& \leq (1+\epsilon)H_G(t)-(1-\epsilon)H_G(t-1)\leq
h_G(t)+2\epsilon,
\end{align*}
and similarly $\hhat_G(t)\geq h_G(t)-2\epsilon$.
\end{proof}
Note that the bound is very weak: since our best generic lower bound is
$h_G(t)\geq 1/n^2$, the relative error with which we known a point $h_G(t)$ is
$2\epsilon n^2$ (which, of course, is pretty useless).

\smallskip\noindent\textbf{Moments.} Evaluation of the moments of $h_G(-)$ poses further
problems. Actually, by Lemma~\ref{lemma:df} we can deduce that
\[
 \sum_t th_G(t) - 2\epsilon D_G\leq\sum_t t\hat h_G(t)\leq \sum_t
 th_G(t) + 2 \epsilon D_G
\]
with confidence $1-2D_G\epsilon$, where $D_G$ is the diameter of $G$, which
implies that the expected value of $\hat h_G(-)$ is an (almost) unbiased estimator of the expected
value of $h_G(-)$. Nonetheless, the bounds we obtain are horrible (and actually
unusable).

The situation for the variance is even worse, as we have to \emph{prove} that we
can use $\Var{\hhat_G}$ as an estimator to $\Var{h_G}$.
Note that for a fixed graph $G$, $H_G$ is a
precise distribution and $\Var{h_G}$ is an \emph{actual number}. Conversely,
$\hhat_G$ (and hence $\Var{\hhat_G}$) is a random variable\footnote{More
precisely, $\hhat_G$ is a sequence of (stochastically dependent) random variables
$\hhat_G(0)$, $\hhat_G(1)$, \dots}. By Theorem~\ref{thm:hhat}, we know that
$\Hhat_G$ is an (almost) unbiased pointwise estimator for $H_G$, and that we
can control its concentration by suitably choosing the number $m$ of
counters. We are going to derive bounds on the approximation of $\Var{h_G}$
using the values of $\Hhat_G(t)$ up to $ \hat D_G$ (i.e., the iteration at which HyperANF stabilises):
\begin{lemma}
Assume that, for every $0\leq t\leq \hat D_G$,  
$\Hhat_G(t)$ is an estimator of $H_G(t)$ with error $\epsilon$ and confidence $1-\delta$; 
then, $\Var{\hhat_G}$ is an estimator of
$\Var{h_G}$ with error
\[
	\epsilon\leq 8 \epsilon\frac{D_G^3}{\Var{h_G}} + 4 \epsilon^2\frac{ D_G^4}{\Var{h_G}}
\]
and confidence $1-(D_G+1)\delta$.
\end{lemma}
\begin{proof}
Assuming error $\epsilon$ on the values of $\hat H$ in $[0\..D_G]$ implies confidence $1-(D_G + 1)\delta$. 
Since $\hat D_G\leq D_G<\infty$, and by definition $\hat h_G(t)=0$ for $t>\hat D_G$ we have ($t$ ranges in $[0\..D_G]$):
\begin{align*}
\Var{\hhat_G} &= \sum_t t^2 \hat h_G(t) -
\Bigl(\sum_t t \hat h_G(t)\Bigr)^2 \\
& \leq \sum_t t^2 \bigl( h_G(t) +
2\epsilon\bigr) - \Bigl(\sum_t t h_G(t)  - 2\epsilon\sum_ t t\Bigr)^2\\
 & \leq \Var{h_G} +2\epsilon\sum_t t^2 + 4\epsilon E[h_G]\sum_t t \\
 & \leq \Var{h_G} + 4\epsilon D_G^2\bigl(D_G + E[h_G]\bigr)\\
& \leq \Var{h_G} + 8\epsilon D_G^3,
\end{align*}
where $E[h_G]$ is the average path length. Similarly
\[
\Var {\hhat_G} \geq \Var{h_G} - 8\epsilon D_G^3 - 4\epsilon^2D_G^4.
\]
Hence the statement.
\end{proof}
The error and confidence we obtain are again unusable, but the lemma proves
that with enough precision and confidence on $\Hhat_G(-)$ we can get precision
and confidence on $\Var{h_G}$.

The results in this section suggests that if computations involve the moments
the only realistic possibility is to resort to parametric statistics to study the
behaviour of the value of interest on a large number of samples. That is, it is
better to compute a large number of relatively low-precision approximate
neighbourhood functions than a small number of high-precision ones, as
from the former the latter are easily computable by averaging, whereas it is
impossible to obtain a large number of samples of derived values from the
latter. As we will see, this approach works surprisingly well.

\section{SPID}

The main purpose of computing aggregated data such as the distance
distribution is that we can try to define indices that express some structural
property of the graph we study, an obvious example being the average distance, or the effective diameter.

One of the main goal of our recent research has
been finding a simple property that clearly distinguishes between social networks
deriving from human interaction (what is usually called a social network in the
strong or proper sense: DBLP, Facebook, etc.) and web-based graphs, which
share several properties of social networks, and as the latter arise from human
activity, but present a visibly different structure.

In this paper we propose for the first time to use the \emph{index of
dispersion} $\sigma^2/\mu$ (a.k.a.~\emph{variance-to-mean ratio}) of the
distance distribution as a measure of the ``webbiness'' of a social
network. We call such an index the \emph{spid (shortest-paths index of dispersion)\footnote{If we
were to follow strictly the terminology used in this paper, this would be the index of dispersion of
the distance distribution, but we guessed that the acronym IDDD would not have
been as as successful.} of $G$}. In particular, networks with a spid larger than one are
to be considered ``web-like'', whereas networks with a spid smaller than one
are to be considered ``properly social''. We recall that a
distribution is called under- or over-dispersed depending on whether its index
of dispersion is smaller or larger than 1, so a network is properly social or
not depending on whether its distance distribution is under- or
over-dispersed.

The intuition behind the spid is that ``properly social'' networks strongly
favour short connections, whereas in the web long connection are not uncommon:
this intuition will be confirmed in Section~\ref{sec:experiments}.

As discussed in the previous section, in theory estimating the spid is an
impossible task, due to the inherent difficulty of evaluating the moments of
$h_G(-)$. In practice, however, the estimate of the spid computed directly on
runs of HyperANF are quite precise. From the actual neighbourhood
function computed for \texttt{cnr-2000} we deduce that the graph spid is $2.49$.
We then ran $100$ iteration of HyperANF with a relative standard deviation of
$9.37$\%, computing for each of them an estimation of the spid; these values approximately follow a normal
distribution of mean $2.489$ and standard deviation $0.9$ (see
Figure~\ref{fig:spidcdf}).
\begin{figure}
\centering
\includegraphics[scale=.30]{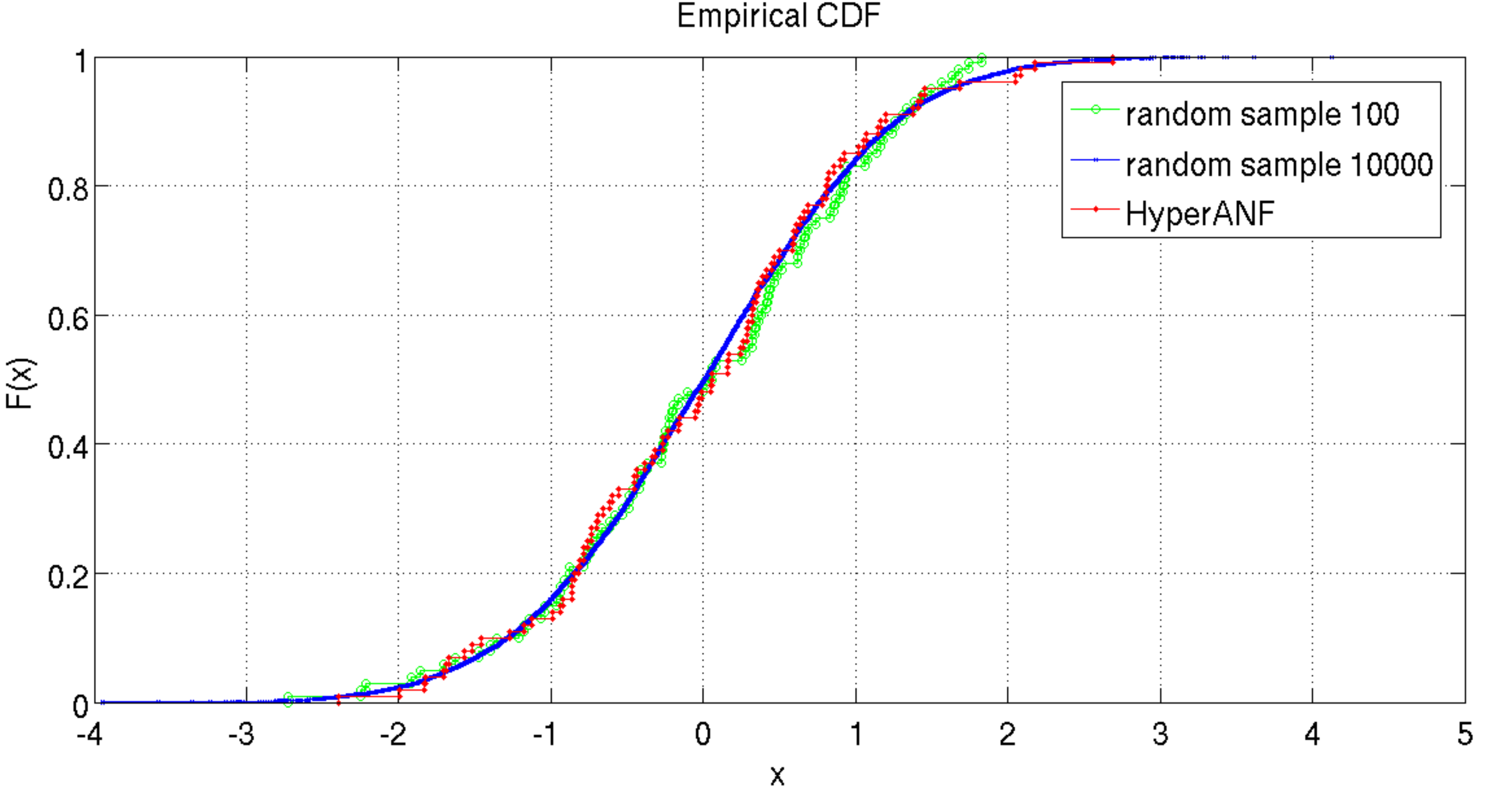}
\caption{\label{fig:spidcdf} Cumulative density function of 100 values of the spid
computed using HyperANF on \texttt{cnr-2000}. For comparison, we also plot
random samples of size $100$ and $10\,000$ drawn from a normal distribution.}
\end{figure}
We obtained analogous concentration results for the average distance. In 
some pathological cases, the distribution is not Gaussian, albeit it
always turns out to be unimodal (in some cases, discarding few outliers), so
we can apply the Vysochanski{\u\i}-Petunin inequality. We will report some
relevant observations on the spid of a number of graphs after
describing our experiments.

\section{Experiments}
\label{sec:experiments}

We ran our experiments on the datasets described in Table~\ref{tab:main}:
\begin{itemize}
\item the web graphs are almost all available at \url{http://law.dsi.unimi.it/},
except for the \texttt{altavista} dataset that was provided by
Yahoo! within the Webscope program (AltaVista webpage connectivity
dataset, version 1.0,
\url{http://research.yahoo.com/Academic_Relations});\footnote{It should be
remarked by this graph, albeit widely used in the literature, is not a good
dataset. The dangling nodes are $53.74$\%---an impossibly
high value~\cite{VigSMCH}, and an almost sure indication that all nodes in the
frontier of the crawler (and not only visited nodes) were added to the graph,
and the giant component is less than 4\% of the whole graph.}
\item for the social
networks: \texttt{hollywood} (\url{http://www.imdb.com/}) is a
co-actorship graph where vertices represent actors; 
\texttt{dblp} (\url{http://www.informatik.uni-trier.de/~ley/db/}) is a
scientific collaboration network where each vertex represents a scientist and
two vertices are connected if they have worked together on an article;
in \texttt{ljournal} (\url{http://www.livejournal.com/}) nodes are
users and there is an arc from $x$ to $y$ if $x$ registered $y$ among his
friends (it is not necessary to ask $y$ permission, so the graph is
\emph{directed});
\texttt{amazon} (\url{http://www.archive.org/details/amazon_similarity_isbn/})
describes similarity among books as reported by the Amazon store; \texttt{enron}
is a partially anonymised corpus of e-mail messages
exchanged by some Enron employees (nodes represent people
and there is an arc from $x$ to $y$ whenever $y$ was the recipient of 
a message sent by $x$);
finally in \texttt{flickr} (\url{http://www.flickr.com/}\footnote{We thank
Yahoo!\ for the experimental results on the Flickr graph.})
vertices correspond to Flickr users and there is an edge connecting $x$ and $y$
whenever either vertex is recorded as a contact of the other one.
\end{itemize}
At the best of our knowledge, this is the first paper where such a wide and
diverse set of data is studied, and where features such as effective diameter or
average path length are computed on very large graphs with precise statistical
guarantees.

\begin{table}	\small
	\centering
	\begin{tabular}{|l|r|r|}
	\hline
	{\bf Graph} & \texttt{snap} & HyperANF\\
	\hline
	\texttt{amazon} & 9.5\,m & 5\,s \\
	\texttt{indochina-2004} & 4.62\,h & 1.83\,m\\
	\texttt{altavista} & - & 1.2\,h\\	
	\hline
	& HADI (90 machines) & HyperANF\\
	\hline
	\begin{minipage}{3cm}Kronecker\\(177\,K nodes, 2\,B arcs)\end{minipage} &
	30\,m & 2.25\,m\\
	\hline 
	\end{tabular}
\caption{\label{tab:speed}A comparison of the speed of \texttt{snap}/HADI vs.~HyperANF.
The tests on \texttt{snap} were performed on our hardware. Both algorithms were
stopped at a relative increment of $0.001$. The timings of HADI on the M45
cluster are the best reported in~\cite{KTAHMRLG}, and both algorithms ran
three iterations. We remark that a run of HyperANF on the Kronecker graph takes
\emph{less than fifteen minutes on a laptop}.}
\end{table}

\begin{landscape}
\begin{table*}[t]
	\centering
	\begin{tabular}{|l|l|r|r|c|c|c|c}
	{\bf Name} & {\bf Type} & {\bf Nodes} & {\bf Arcs} & \textbf{spid} $(\pm\sigma)$ & \textbf{ad} $(\pm\sigma)$ & \textbf{ied} $(\pm\sigma)$ & \textbf{ed (2)}\\
	\hline
        \texttt{amazon} & social (u) & $735\,323$  & $5\,158\,388$             	& $0.76\,(\pm0.060)$ &    $12.05\,(\pm0.206)$&     $15.50\,(\pm0.433)$	& $[14\..18]$ \\
        \texttt{dblp} & social (u) & $326\,186$  & $1\,615\,400$               	& $0.36\,(\pm0.034)$ &    $7.34\,(\pm0.114)$&     $8.96\,(\pm0.215)$	& $[8\..10]$ \\
        \texttt{enron} & social (d) & $69\,244$  & $276\,143$                  	& $0.21\,(\pm0.020)$ &    $4.24\,(\pm0.065)$&     $4.94\,(\pm0.103)$	& $[4\..6]$ \\
        \texttt{ljournal} & social (d) & $5\,363\,260$  & $79\,023\,142$       	& $0.21\,(\pm0.023)$ &    $5.99\,(\pm0.078)$&     $6.92\,(\pm0.143)$	& $[6\..8]$ \\
        \texttt{flickr} & social (u) & $526\,606$  & $47\,097\,454$            	& $0.14\,(\pm0.009)$ &    $3.50\,(\pm0.047)$&     $3.92\,(\pm0.049)$	& $[3\..5]$ \\
        \texttt{hollywood} & social (u) & $1\,139\,905$  & $113\,891\,327$     	& $0.14\,(\pm0.012)$ &    $3.87\,(\pm0.045)$&     $4.42\,(\pm0.109)$	& $[4\..5]$ \\
        \texttt{indochina-2004-hosts} & host (d) & $19\,123$  & $233\,380$     	& $0.35\,(\pm0.021)$ &    $4.26\,(\pm0.079)$&     $5.44\,(\pm0.164)$	& $[5\..7]$ \\
        \texttt{uk-2005-hosts} & host (d) & $587\,205$  & $12\,825\,465$       	& $0.30\,(\pm0.018)$ &    $5.93\,(\pm0.081)$&     $7.06\,(\pm0.151)$	& $[6\..8]$ \\
        \texttt{cnr-2000} & web (d) & $325\,557$ & $3\,216\,152$               	& $2.50\,(\pm0.086)$ &    $17.35\,(\pm0.313)$&     $25.45\,(\pm0.357)$	& $[23\..29]$ \\
        \texttt{eu-2005} & web (d) & $862\,664$  & $19\,235\,140$              	& $1.25\,(\pm0.209)$ &    $10.17\,(\pm0.363)$&     $14.31\,(\pm0.988)$	& $[13\..16]$ \\
        \texttt{in-2004} & web (d) & $1\,382\,908$  & $16\,917\,053$           	& $1.30\,(\pm0.173)$ &    $15.40\,(\pm0.374)$&     $20.74\,(\pm0.792)$	& $[20\..24]$ \\
        \texttt{indochina-2004} & web (d) & $7\,414\,866$  & $194\,109\,311$   	& $1.64\,(\pm0.134)$ &    $15.63\,(\pm0.338)$&     $21.68\,(\pm0.658)$	& $[20\..26]$ \\
        \texttt{uk@10E6} & web (d) & $100\,000$  & $3\,050\,615$               	& $1.64\,(\pm0.111)$ &    $5.97\,(\pm0.172)$&     $10.36\,(\pm0.251)$	& $[8\..12]$ \\
        \texttt{uk@10E7} & web (d) & $1\,000\,000$  & $41\,247\,159$           	& $1.76\,(\pm0.043)$ &    $8.96\,(\pm0.172)$&     $14.31\,(\pm0.341)$	& $[12\..17]$ \\
        \texttt{it-2004} & web (d) & $41\,291\,594$  & $1\,150\,725\,436$      	& $2.14\,(\pm0.149)$ &    $15.02\,(\pm0.300)$&     $19.65\,(\pm0.698)$  & $[18\..22]$ \\
        \texttt{uk-2007-05} & web (d) & $105\,896\,555$  & $3\,738\,733\,648$  	& $1.10\,(\pm0.234)$ &    $15.39\,(\pm0.418)$&     $19.93\,(\pm1.030)$  & $[18\..23]$ \\
        \texttt{altavista} & web (d) & $1\,413\,511\,390$ & $6\,636\,600\,779$ 	& $4.24\,(\pm0.764)$ &    $16.69\,(\pm0.779)$&     $23.04\,(\pm2.517)$	& $[19\..31]$ \\
	\end{tabular}
\caption{\label{tab:main}Our main data table. ``Type'' describes
whether the given graph is a web-graph, a proper social network, or the host
quotient of a web graph (u=undirected, d=directed). The graphs \texttt{uk@10E6}
and \texttt{uk@10E7} are obtained by visiting in a breadth-first fashion
\texttt{uk-2007-05} starting from a random node. They simulate smaller crawls
of a larger network. We show spid, average distance and interpolated effective diameter
\textit{a posteriori} with their empirical standard deviation, and intervals for the effective diameter
with $85$\% confidence for a comparison. 
}
\end{table*}
\end{landscape}

All experiments are performed on a Linux server 
equipped with Intel Xeon X5660 CPUs ($2.80$\,GHz, $12$\,MB cache
size) for overall 24 cores and $128$\,GB of RAM; the server cost about
$8\,900$ EUR in 2010.

A brief comparison with \texttt{snap} and HADI timings is shown in
Table~\ref{tab:speed}. Essentially, on our hardware HyperANF is two orders of
magnitudes faster than \texttt{snap}. Our run on the Kronecker graph is one order of
magnitude faster than HADI's (or three orders of magnitude faster, if you take into consideration
the number of machines involved), but this comparison is unfair, as in principle
HADI can scale to arbitrarily large graphs, whereas we are limited by the amount
of memory available. Nonetheless, the speedup is clearly a breakthrough in the
analysis of large graphs. It would be interesting to compare our timings for
the \texttt{altavista} dataset with HADI's, but none have been published.

It is this speed that makes it possible, for the first time, to compute data
associated with the distance distribution with high precision and for a
large number of graphs. We have 100 runs with relative standard
deviation of $9.37$\% for all graphs, except for those on the \texttt{altavista}
dataset ($13.25$\%). All graphs are run to stabilisation.
Our computations are necessarily much longer (usually, an order of magnitude longer
in iterations) than those used to compute the effective diameter or similar measures.
This is due to the necessity of computing with high precision second-order statistics
that are used to compute the spid.

Previous publications used few graphs, mainly because of the large
computational effort that was necessary, and no data was available about the
number of runs. Moreover, we give precise confidence intervals based on
parametric statistics for data depending on the second moment, such as the
spid---something that has never done before. We gather here our findings.

\smallskip\noindent\textbf{\textit{A posteriori} parameters are highly
concentrated.} According to our experiments, computing the effective diameter,
average distance and spid on a large number of low-precision runs generates highly concentrated
distributions (see the empirical standard deviation in Table~\ref{tab:main}). Thus,
we suggest this approach for computing such values, \emph{provided that
termination is by stabilisation}.

\begin{figure}[htb]
\centering
\includegraphics[scale=.7]{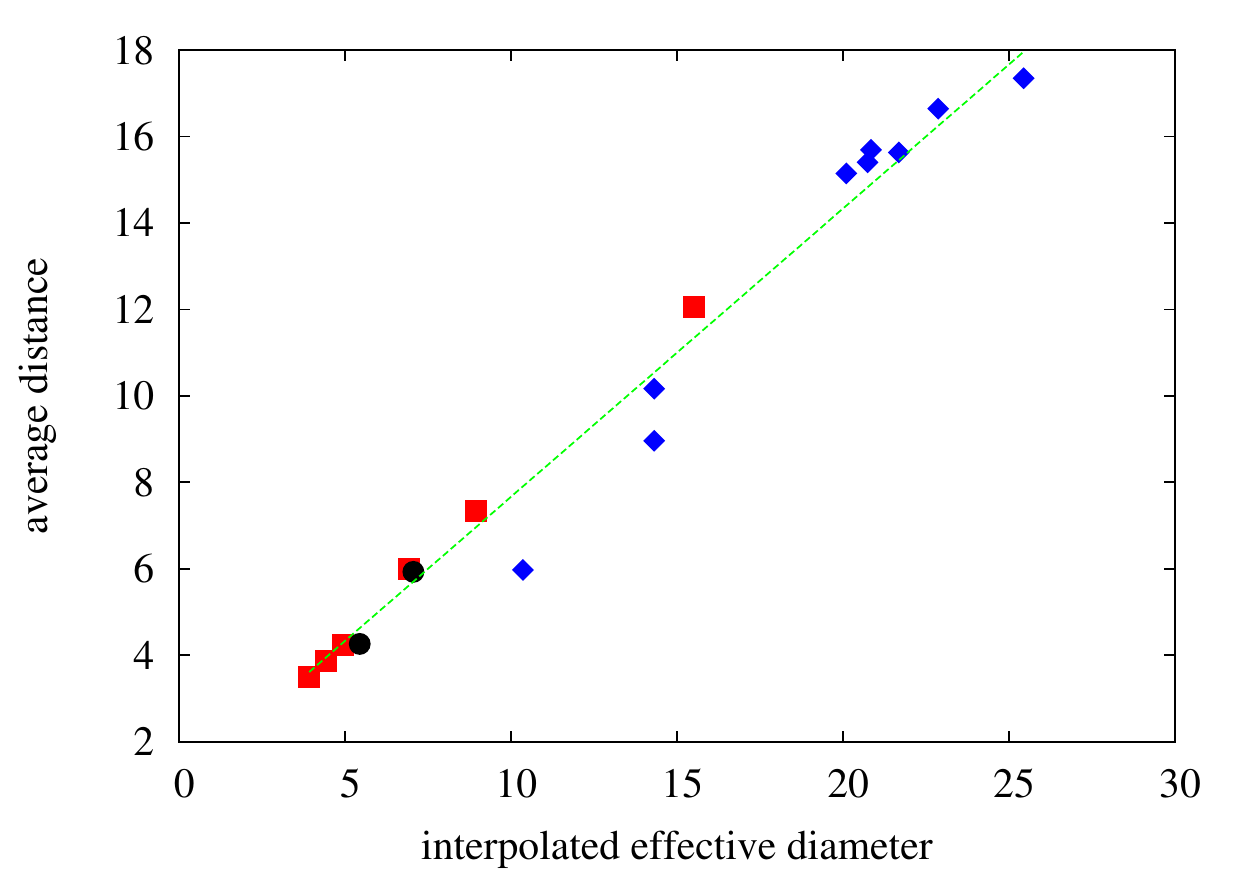}
\caption{\label{fig:diamvsavgsp} A plot showing the strong linear correlation
between the average distance and the effective diameter.}
\end{figure}

\smallskip\noindent\textbf{Effective diameter and average distance are
essentially linearly correlated.} Figure~\ref{fig:diamvsavgsp} shows a scatter plot
of the two values, and the line $2x/3+1$. The correlation between the two values
has always been folklore in the study of social networks, but we can confirm
that on both social and web networks the connection can be exactly expressed in
linear terms (it would be of course interesting to prove such a correlation
formally, under suitable restrictions on the structure of the graph). This fact
suggests that the average distance (which is more principled from a
statistic viewpoint, and parameter-free) should be used as the reference
parameter to express the closeness between nodes. Moreover, experimentally \emph{the standard deviation of the
effective diameter in a posteriori computations is usually significantly larger
than that of the average distance}.

Incidentally, the average distance of the \texttt{altavista} dataset is $16.5$---slightly more than
what reported in~\cite{KTAHMRLG}  (possibly because of termination conditions
artifacts).

\smallskip\noindent\textbf{It is difficult to give \textit{a priori} confidence
intervals for the effective diameter with a small number of runs.} Unless a
large number of runs is available, so that the precision of the approximation of
the neighbourhood function can be significantly lowered, it is impossible to
provide interesting upper bounds for the effective diameter.

\smallskip\noindent\textbf{The spid can tell social networks from web graphs.}
As shown in Table~\ref{tab:main}, even taking the
standard deviation into account spids are pretty much below 1 for social
networks and above 1 for web graphs; host graphs (not surprisingly) behave
like social networks. Note that this works \emph{both for directed and undirected graphs}.
Figure~\ref{fig:spidsize} shows the spid values obtained
for our datasets plotted against the graph size, and also witnesses that there
is no correlation (a similar graph, not shown here, testifies that
spid is also independent from density). Figure~\ref{fig:spidavgsp} shows that there is some
slight correlation between the spid and the average distance: nonetheless,
there is no way to tell networks from our dataset apart using the latter value,
whereas the under- or over-dispersion of the distance distribution, as defined by the
spid, never makes a mistake. Of course, we expect to enrich this graph in
time with more datasets: we are particularly interested in gathering very large
social networks to test the spid at large sizes.

We remark that, as a sanity check, we have also computed on several
web-graph datasets the spid of the \emph{giant component}, which turned out to
be very similar to the spid of the whole graph. We see this as a clear sign that
the spid \emph{is largely independent of the artifacts of the crawling process}.

\smallskip\noindent\textbf{Direction should not be destroyed when analysing a
graph.} We confirm that symmetrising graphs destroys the combinatorial structure
of the network: the average distance drops to very low values in
all cases, as well as the spid. This suggests that there is important structural
information that is being ignored. We also note that since all web snapshot we
have at hand are gathered by some kind of breadth-first visit, they represent
balls of small diameter centred at the seed: symmetrising the graph we cannot
expect to get an average distance that is larger than twice the radius of
the ball. All in all, the only advantage of symmetrising a graph is a significant reduction
in the number of iterations that are needed to complete a computation of the
neighbourhood function.\footnote{We remark that the ``diameter $7\sim8$'' claim
in~\cite{KTAHMRLG} about the \texttt{altavista} dataset
refers to the \emph{effective} diameter
for the \emph{symmetrised} version of the graph.}

\begin{figure}[htb]
\centering
\includegraphics[scale=.7]{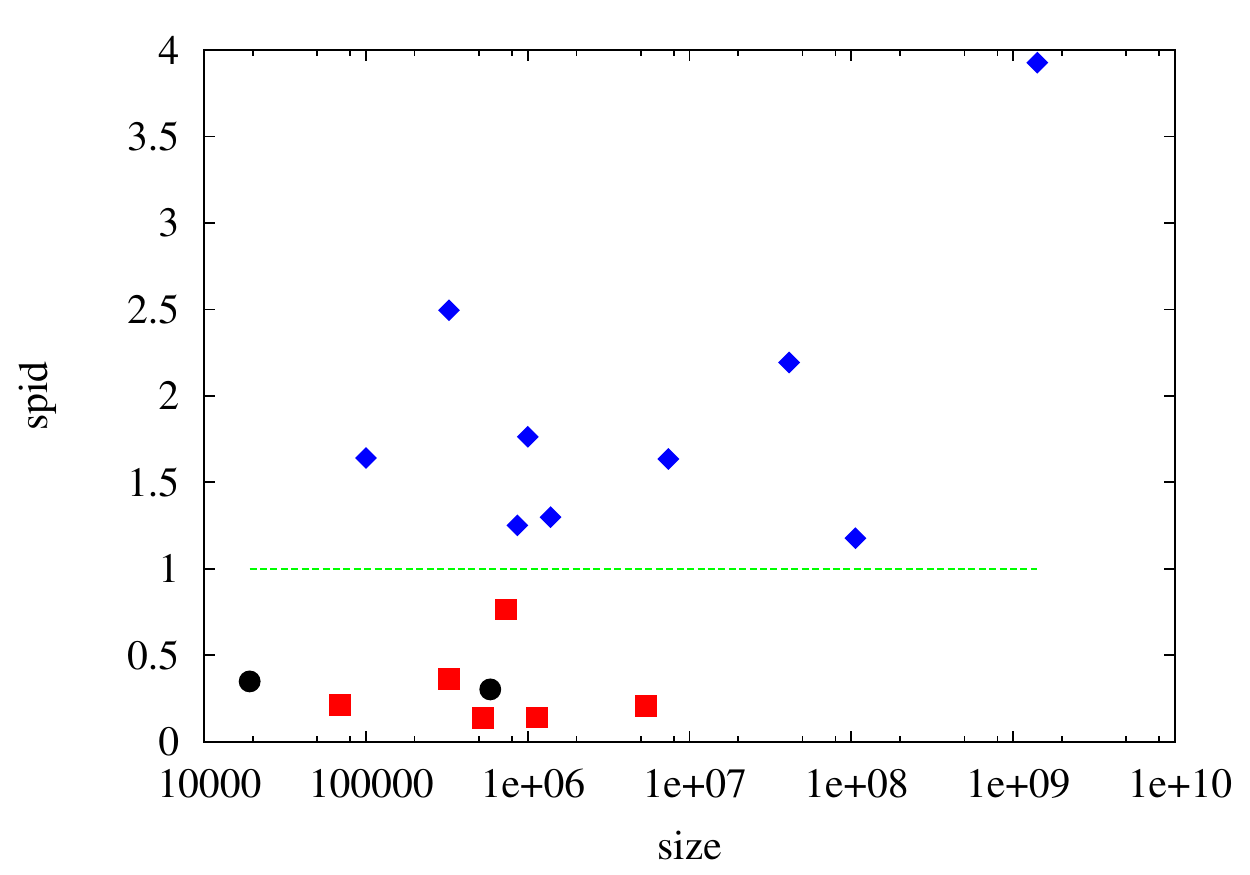}
\caption{\label{fig:spidsize} A plot showing the spid values (vertical) for our
datasets compared with their size (i.e., number of nodes, horizontal): red
squares correspond to social networks, blue diamonds to web graphs and black
circles to host graphs.}
\end{figure}

\begin{figure}[htb]
\centering
\includegraphics[scale=.7]{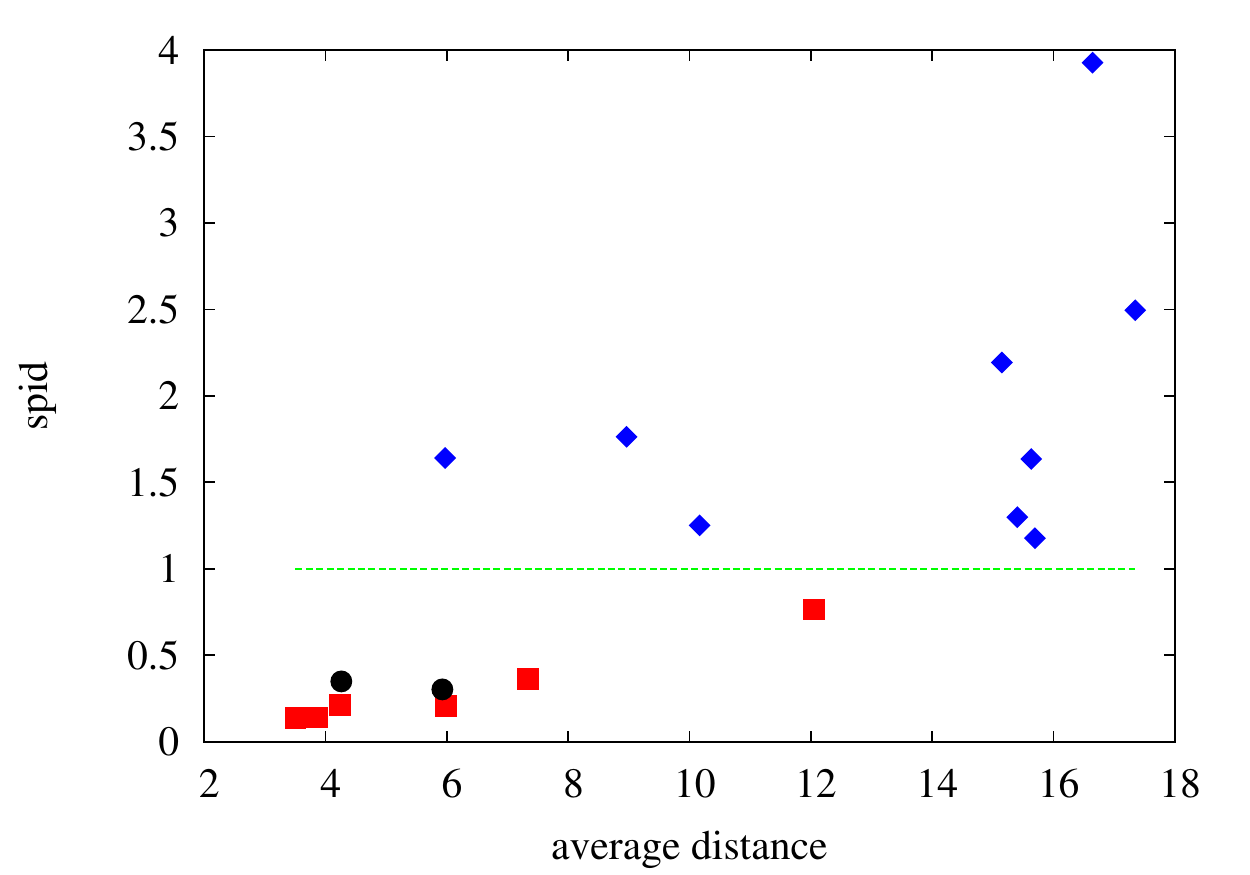}
\caption{\label{fig:spidavgsp} A plot showing the spid against the average
distance using the same conventions of Figure~\ref{fig:spidsize}.}
\end{figure}

To give a more direct idea of the level of
precision of our diameter estimation, we computed the actual diameter at $\alpha$ for the \texttt{cnr-2000} dataset,
and plotted it against the interval estimation obtained by HyperANF

\begin{figure}
\centering
\includegraphics[scale=.7]{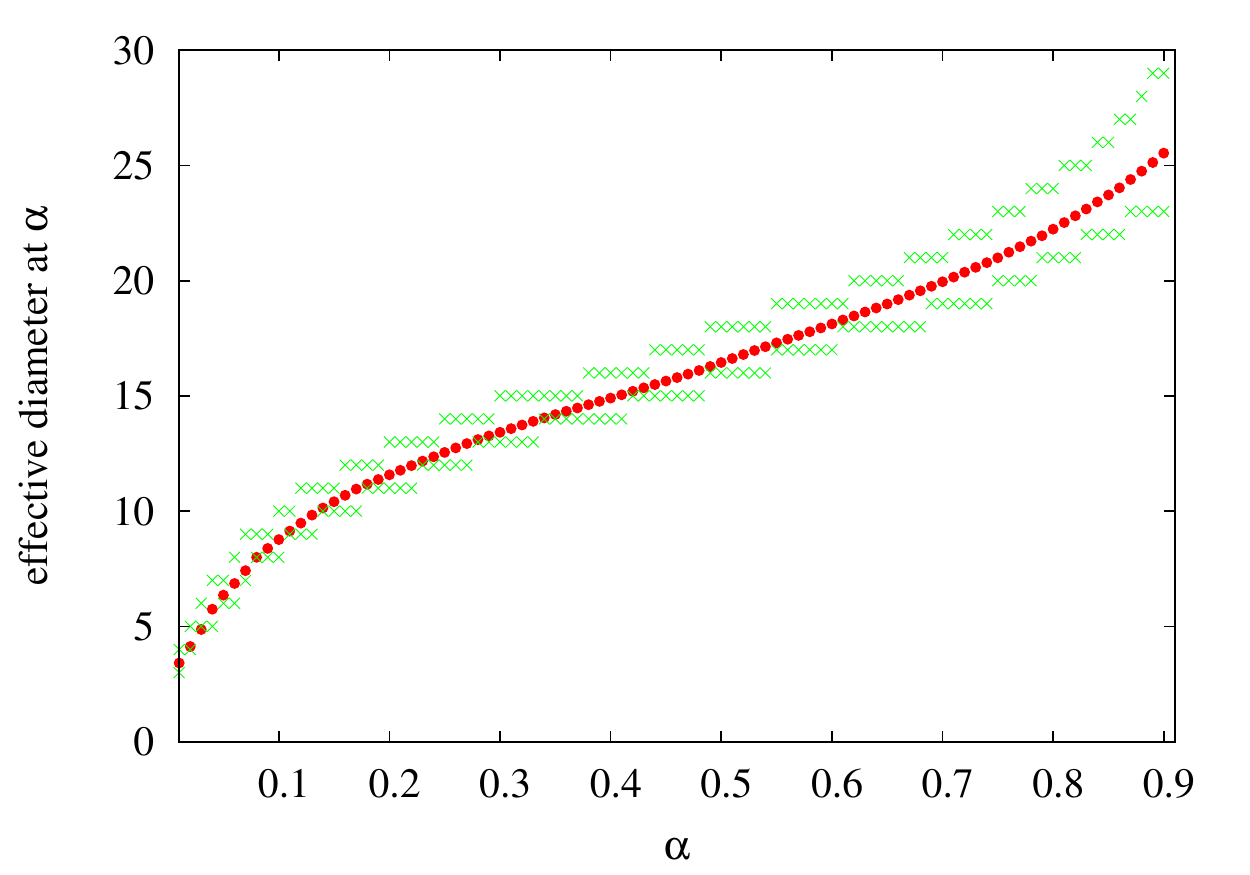}
\caption{\label{fig:diam} Effective diameters at $\alpha$ for the
\texttt{cnr-2000} dataset; red bullets show the real effective diameter, whereas green
crosses show the upper and lower extreme of the confidence interval obtained
running $100$ HyperANF with $m=128$.}
\end{figure}


\section{Future work}

HyperANF lends itself naturally to distributed implementations. However,
contrarily to the approach taken by HADI \cite{KTAHMRLG}, we think that the
correct parallel framework for implementing a diffusing computation is a
synchronous parallel system where computation happens at nodes and communication
is sent from node to node with messages. Such a framework, Pregel, has been
recently developed at Google~\cite{MABP}. In a Pregel implementation of
HyperANF, every computational node sends its own counter as message to its
predecessors if it changed from the previous iteration, waits for incoming messages from its
successors, and computes the maximisation procedure on the received messages.
Due to the small size of HyperLogLog counter (exponentially smaller than the
Flajolet--Martin counters used by ANF), the amount of communication would be
very small.

Although in this
paper, we preferred to focus on the computation of the spid, we remark that
HyperANF can also be used to build the radius distribution described in~\cite{KTAHMRLG}, 
or the related closeness centrality. 

\section{Conclusions}

HyperANF is a breakthrough improvement over the original ANF techniques, mainly
because of the usage of the more powerful HyperLogLog counters combined with
their fast broadword combination and systolic computation. HyperANF can run
to stabilisation very large graphs, computing data with statistical guarantees.


We consider, however, the introduction of the spid of a graph the main 
conceptual contribution of this paper. HyperLogLog is instrumental in making the
computation of the spid possible, as the latter requires a number of iterations
that is an order of magnitude larger than those required for an estimate of the
effective diameter.

\paragraph{Acknowledgements} Flavio Chierichetti participated to the earlier
phases of this work. We want to thank Dario Malchiodi for fruitful discussions and hints.

\bibliography{biblio}
\end{document}